\documentclass[12pt]{article}
\usepackage{amsfonts}
\usepackage{amsmath}
\usepackage{amsthm}
\usepackage{cite}
\newtheorem{theorem}{Theorem}[section]
\newtheorem{lemma}[theorem]{Lemma}
\newtheorem{proposition}[theorem]{Proposition}
\newtheorem{corollary}[theorem]{Corollary}

\newtheorem{remark}[theorem]{Remark}

\usepackage{graphicx}
\numberwithin{equation}{section}

\def\ssl{{\mathfrak{sl}}}
\def\bbbt{{\mathbb T}}
\def\bbbn{{\mathbb N}}
\def\bbbc{{\mathbb C}}
\def\bbbd{{\mathbb D}}
\def\bbbz{{\mathbb Z}}

\def\be{{\mathbf e}}

\def\cB{{\cal B}}

\def\cR{{\cal R}}

\begin{document}

\title{Darboux transformations with tetrahedral reduction group and related 
integrable systems.}
\author{George Berkeley$^{\star}$, Alexander V. Mikhailov$^{\star}$
and Pavlos Xenitidis $ ^\dagger $\\
$\star$ Department of Applied Mathematics , University of Leeds, UK\\
$\dagger$ School of Mathematics, Statistics $\&$ Actuarial Science,\\ 
University of Kent, UK }
\date{}
\maketitle

\begin{abstract}
In this paper we derive new two-component integrable differential difference and partial difference systems by applying a Lax-Darboux scheme to an operator formed from an $\ssl_3(\bbbc)$-based automorphic Lie algebra. The integrability of the found systems is demonstrated via Lax pairs and generalised symmetries.  
\end{abstract}

\section{Introduction}

Most interesting and useful integrable systems in applications are often  
the result of a reduction of a generic system. A systematic study of reductions 
originating from the works  \cite{mik79,mik80,mik81,mik_dis} led 
to the concept of the reduction group. It had been realised that
reductions of the Lax structures can be associated with its invariance with
respect to a group of automorphisms ({\sl  a reduction group}), which include a fractional
linear 
transformation of the spectral parameter and simultaneous Lie algebra
automorphisms. Subalgebras of the generalised loop Lie algebra, which are 
invariant with
respect to reduction groups were introduced in \cite{mik_dis}. 
These subalgebras  have been further studied in \cite{LM05} where they  
acquired the name {\sl automorphic Lie algebras}. In this setup the problem of 
classification of 
(algebraic) reductions can be reduced to the problem of classification of the 
reduction groups, or more precisely to the problem of classification of 
automorphic Lie algebras. 

In the case of a Lax operator which is rational with respect to the spectral parameter, a 
reduction group is finite and isomorphic to one of the following list: a cyclic 
group, a dihedral group, a symmetry group of a Platonic solid or a direct product of a 
finite number of those. There is a good
progress in the problem of classification of automorphic Lie algebras 
corresponding to finite reduction groups acting by inner automorphisms on the 
base Lie algebra \cite{LM05,LM04,sara,bm2,bury,ls10,kls14,kls15}. In 
particular, it has been shown that in the case of  $\ssl_2$ based algebras 
it is sufficient to consider reduction groups isomorphic to $\bbbz_2$ and 
$\bbbz_2\times \bbbz_2\simeq \bbbd_2$ (other finite reduction groups result in 
the same set of automorphic Lie algebras). The Lax structures and integrable 
partial differential equations, corresponding to automorphic Lie algebras, have 
been studied in  \cite{mik79,mik80,mik81,LM04,bm2,bury}. Having a Lax structure we can construct 
Darboux transformations which lead us to associated differential-difference and 
finite difference integrable systems, i.e. extend it to a Lax-Darboux scheme.

Recently we 
extended the reduction group method to the Lax-Darboux schemes of 
Nonlinear Schr\"odinger type, associated with  all $\ssl_2$ based 
automorphic Lie algebras \cite{smx}. We constructed reduction group invariant 
Darboux matrices and corresponding differential-difference and finite 
difference integrable systems.  For $\ssl_N,\ N>2$, the irreducible 
faithful projective representations of finite reduction groups 
are possible only 
in the cases $N=3,4,5,6$. In this paper we construct the extension of the Lax 
Darboux scheme for the case $N=3$ with the tetrahedral reduction group and the 
poles of the Lax operator located at the points of the most degenerate orbit 
(the vertices of the tetrahedron).  The hyperbolic and parabolic partial 
differential equations associated with such Lax operator have been studied in 
\cite{mik80} and \cite{mr89g:58092} respectively. Here we construct their corresponding Darboux matrices and B{\"a}cklund transformations. We interpret the Darboux matrices as defining shifts on the lattice and the corresponding B{\"a}cklund transformations as non evolutionary differential-difference equations. In this interpretation we also derive new systems of differential-difference equations, as well as new systems of difference equations. Moreover, we present reductions, potentiations and Miura transformations which relate our discrete systems to Bogoyavlensky lattices and some new six-point difference equations.

It can be shown that for the $N=3$ case, with the most degenerate orbit, Lax
pairs corresponding to automorphic Lie algebras with tetrahedral, octahedral and
icosahedral reduction group produce integrable PDEs which are related by point
transformation \cite{bury}.  This is consistent with the recent 
results published in \cite{kls15}  where good progress in 
the problem of classification of automorphic Lie algebras with inner 
automorphisms has been achieved. Hence the case we study can be considered as
representative of Lax pairs with finite reduction group and such an orbit.

The paper is organised as follows. The next section contains all the necessary background material on the tetrahedral reduction group. In Section \ref{deriveDarbouxs} we consider Lax operators invariant under this reduction group and derive corresponding Darboux transformations along with integrable non evolutionary differential difference equations. In Section  \ref{localSymDerivation} we employ the Darboux transformations and systematically construct evolutionary differential-difference  systems. Section \ref{discretesysDerivations} deals with the derivation of some new fully discrete systems and the next one contains related integrable systems, obtained via reductions, potentiations and Miura transformations. The last section presents an overall evaluation of the results obtained in the main body of the paper.

\section{Tetrahedral reduction group and corresponding automorphic Lie algebra}

The tetrahedral reduction group was introduced and studied in 
\cite{mik80} and \cite{mr89g:58092} in the context of integrable partial 
differential equations and their Lax representations.

The group of orientation-preserving symmetries of a regular tetrahedron, known 
as the tetrahedral group has 12 elements and it is isomorphic to the 
alternating group $A_4$, which can be generated by two 
elements $r,s$ satisfying relations
\begin{equation}
A_4=<r,s\ |\ r^2=s^3=(rs)^3=e>.
\end{equation}

We consider two projective faithful representations of the 
tetrahedral group. It is sufficient to define the representation for the 
generators of the group. The first one $\sigma:A_4\mapsto PSL_2(\bbbc)$
we realise by fractional linear transformations 
\begin{equation}
 \label{sigma}
 \sigma_s:\lambda\mapsto\omega\lambda,\quad 
\sigma_r:\lambda\mapsto\frac{\lambda +2}{\lambda -1},\qquad \omega 
=\exp\left(\frac{2\pi i}{3}\right).
\end{equation}
The second representation
$\rho: A_4\mapsto PSL_3(\bbbc)$ is given by inner automorphisms of the algebra 
$\ssl_3(\bbbc)$
\[
 \rho_s:{\bf a}\mapsto Q_s{\bf a}Q_s^{-1},\quad \rho_r:{\bf a}\mapsto Q_r{\bf 
a}Q_r^{-1},\qquad {\bf a}\in \ssl_3(\bbbc),
\]
where 
\begin{displaymath}\textbf{Q}_s=\left(
\begin{array}{ccc}
\omega &  0 & 0 \\
0 & \omega^2 & 0 \\
0 & 0 & 1 
\end{array} \right),
\textbf{Q}_r=\left(
\begin{array}{ccc}
-1 & 2 & 2 \\
2 & -1 & 2 \\
2 & 2 & -1 
\end{array} \right).  \end{displaymath} 
It is easy to verify that representations $\sigma$ and $\rho$ are faithful and 
irreducible. The group of linear fractional transformations generated by 
$\sigma_s,\sigma_r$ is a subgroup of the group of automorphisms of the field of 
rational functions $\bbbc(\lambda)$ of variable $\lambda$, while  the group 
generated by $\rho_r,\rho_s$ is a subgroup of inner automorphisms of the 
Lie algebra $\ssl_3(\bbbc)$. 

We consider a generalised loop algebra $\ssl_3(\bbbc(\lambda))$ 
and transformations
\[
 g_r: {\bf a}(\lambda)\mapsto Q_r({\bf a}(\sigma_r^{-1}(\lambda)))Q_r^{-1},\  g_s: 
{\bf a}(\lambda)\mapsto Q_s({\bf{a}}(\sigma_s^{-1}(\lambda)))Q_s^{-1},\quad {\bf 
a}(\lambda)\in \ssl_3(\bbbc(\lambda)).
\]
This group  $\bbbt$ generated by $g_r,g_s$  is isomorphic to the alternating
group $A_4$ and is a subgroup of the group of automorphisms of the algebra 
$\ssl_3(\bbbc(\lambda))$. Following \cite{mik80}, \cite{mr89g:58092} we shall 
refer to $\bbbt$ as the tetrahedral reduction group. The $\bbbt$ invariant subalgebra 
\[
  \ssl_3(\bbbc(\lambda))^\bbbt=\{ {\bf 
a}(\lambda)\in \ssl_3(\bbbc(\lambda))\,|\, g_s({\bf 
a}(\lambda))=  g_r({\bf 
a}(\lambda))={\bf 
a}(\lambda)\}
\]
is called a tetrahedral automorphic Lie algebra.

The group of fractional linear transformations generated by 
$\sigma_r,\sigma_s$ (\ref{sigma}) has three degenerate orbits corresponding to 
fixed points. Two of them are  
vertices of the tetrahedron and the dual tetrahedron 
\[
 \bbbt(\infty)=\{\infty, 1,\omega,\omega^2\}, \qquad  \bbbt(0)=\{0, 
-2,-2\omega,-2\omega^2\}
\]
and one orbit corresponds to the middles of the edges
\[
 \bbbt(1+\sqrt{3})=\{\lambda\in\bbbc\,|\, \lambda^6-20\lambda^3-8=0\}\, .
\]
Let us define the automorphic subalgebra $\mathfrak{A}(\infty)\subset  
\ssl_3(\bbbc(\lambda))^\bbbt$ corresponding to the orbit $ \bbbt(\infty)$. Let 
$\cR(\bbbt(\infty))$ denote the set of all rational functions which may have 
poles at the points of the orbit $\bbbt(\infty)$ and are regular elsewhere. 
Obviously  $\cR(\bbbt(\infty))$ is a ring with automorphisms generated by 
$\sigma_r,\sigma_s$ (\ref{sigma}). Then
\[
 \mathfrak{A}(\infty)= \{{\bf 
a}(\lambda)\in  \cR(\bbbt(\infty))\bigotimes 
\ssl_3(\bbbc)\,|\,   g_s({\bf 
a}(\lambda))= g_r({\bf 
a}(\lambda))={\bf 
a}(\lambda)\}.  
\]
Similarly one can define subalgebras $\mathfrak{A}(\Gamma)\subset  
\ssl_3(\bbbc(\lambda))^\bbbt$, where $\Gamma\subset\bbbc$ is any finite set of 
points 
\[
 \mathfrak{A}(\Gamma)= \cR(\bbbt(\Gamma))\bigotimes 
\ssl_3(\bbbc)\bigcap \ssl_3(\bbbc(\lambda))^\bbbt.
\]

The generators of $ \mathfrak{A}(\infty)$ can be constructed as
\begin{equation} \mathbf{e}_1=\langle\lambda 
\mathbf{E}_{13}\rangle_\mathbb{T},\enspace{} 
\mathbf{e}_2=\langle\lambda \textbf{E}_{21}\rangle_\mathbb{T},\enspace{} 
\mathbf{e}_3=\langle\lambda \mathbf{E}_{32}\rangle_\mathbb{T},\end{equation}
where $\mathbf{E}_{ij}$ is the matrix with elements 
$[\mathbf{E}_{i,j}]_{kl}=\delta_{ik}\delta_{jl}$ and $\langle a(\lambda)\rangle 
_\mathbb{T} $ denotes  the group 
average operator
\begin{equation}\langle a(\lambda)\rangle_\mathbb{T}=\frac{1}{12}\sum_{g\in 
\mathbb{T}}g(a(\lambda)).\end{equation}
These are given explicitly as
\begin{displaymath}\textbf{e}_1=\left(
\begin{array}{ccc}
-\frac{1}{6}\frac{\lambda^3+2}{\lambda^3-1} &  -\frac{1}{2}\frac{\lambda^2}{\lambda^3-1} & \frac{1}{4}\frac{\lambda^4}{\lambda^3-1} \\
\frac{\lambda}{\lambda^3-1} & \frac{1}{3}\frac{\lambda^3+2}{\lambda^3-1} & -\frac{1}{2}\frac{\lambda^2}{\lambda^3-1} \\
\frac{\lambda^3+2}{\lambda^3-1} & \frac{\lambda }{\lambda^3-1} & -\frac{1}{6}\frac{\lambda^3+2}{\lambda^3-1} 
\end{array} \right)\end{displaymath}

\begin{displaymath}\textbf{e}_2=\left(
\begin{array}{ccc}
-\frac{1}{6}\frac{\lambda^3+2}{\lambda^3-1} & \frac{\lambda^2}{\lambda^3-1} & \frac{\lambda }{\lambda^3-1} \\
\frac{1}{4}\frac{\lambda^4}{\lambda^3-1} & -\frac{1}{6}\frac{\lambda^3+2}{\lambda^3-1} & -\frac{1}{2}\frac{\lambda^2}{\lambda^3-1} \\
-\frac{1}{2}\frac{\lambda^2 }{\lambda^3-1} & \frac{\lambda }{\lambda^3-1} & \frac{1}{3}\frac{\lambda^3+2}{\lambda^3-1} 
\end{array} \right)\end{displaymath}

\begin{displaymath}\textbf{e}_3=\left(
\begin{array}{ccc}
\frac{1}{3}\frac{\lambda^3+2}{\lambda^3-1} & -\frac{1}{2}\frac{\lambda^2}{\lambda^3-1} & \frac{\lambda }{\lambda^3-1} \\
\frac{\lambda }{\lambda^3-1} & -\frac{1}{6}\frac{\lambda^3+2}{\lambda^3-1} & \frac{\lambda^2}{\lambda^3-1} \\
-\frac{1}{2}\frac{\lambda^2 }{\lambda^3-1} & \frac{1}{4}\frac{\lambda^4 }{\lambda^3-1} & -\frac{1}{6}\frac{\lambda^3+2}{\lambda^3-1} 
\end{array} \right).\end{displaymath}

In fact it can be shown that $\mathfrak{A}(\Gamma)$ 
possesses a quasi graded structure
\[ \mathfrak{A}(\Gamma)=\bigoplus_{k=1}^\infty A^k,\qquad [A^n,A^m]\subset 
A^{n+m}\bigoplus A^{n+m-1}.\] 
In the case $\mathfrak{A}(\infty)$
\[A^k=\{J^{k-1}\mathbf{e}_1,J^{k-1}\mathbf{e}_2,\ldots 
,J^{k-1}\mathbf{e}_8\},\qquad k\in\bbbn,\]
where 
\begin{equation}\begin{aligned}\label{basis}
& \mathbf{e}_4=[\mathbf{e}_1,\mathbf{e}_3],\text{  } 
\mathbf{e}_5=[\mathbf{e}_2,\mathbf{e}_1],\text{  } 
\mathbf{e}_6=[\mathbf{e}_3,\mathbf{e}_2]\\
& \mathbf{e}_7=[[\mathbf{e}_1,\mathbf{e}_3],\mathbf{e}_2],\text{  
}\mathbf{e}_8=[[\mathbf{e}_2,\mathbf{e}_1],\mathbf{e}_3]
\end{aligned}
\end{equation} 
 and $J$ is the automorphic function of $\lambda$
 \[
  J\,=\,\frac{\lambda ^3 \left(\lambda ^3+8\right)^3}{4 
\left(\lambda ^3-1\right)^3.}
 \]
An automorphic Lie algebra can also be regarded as a finite dimensional $\bbbc[J]$  
module with Lie product, which is a special degenerate case of polynomial Lie 
algebras introduced in \cite{BL2002}. This proved to be convenient in the 
problem of classification of automorphic Lie algebras 
\cite{bm2,bury,kls14,kls15}.

It is also useful to define an automorphic associative algebra 
$\mathfrak{B}(\infty)$ as the algebra generated by matrices $\be_1,\ldots ,\be_8$ 
and the identity matrix $\mathbf{I}$. Indeed, we can take products  
$\be_n \be_m$ with usual matrix multiplication and the result can be 
represented as a linear combination of $\be_k,\ k=1,\ldots,8$ and the identity 
matrix with coefficients from $\bbbc[J]$
\[
 \be_n \be_m=\sum_{k=1}^8 B_{nm}^k(J)\be_k+D_{nm}(J)\mathbf{I},
\]
where the structure constants $B_{nm}^k(J)$ are linear functions of $J$ and 
$D_{nm}(J)=D_{mn}(J)$ are quadratic polynomials of the automorphic function 
$J$. Thus, algebra $\mathfrak{B}(\infty)$ can be regarded either as a finitely 
generated  $\bbbc[J]$ module with associative multiplication, or as an infinite 
dimensional quasi-graded associative algebra 
\[
 \mathfrak{B}(\infty)=\bigoplus_{k=1}^\infty {\cB}_k,\qquad 
{\cB}_n{\cB}_m\subset {\cB}_{n+m+1}\bigoplus{\cB}_{n+m}\bigoplus{\cB}_{n+m-1},
\]
where
\[
 {\cB}_k=J^{k-1}{\rm span}_\bbbc(\mathbf{I}, 
\mathbf{e}_1,\mathbf{e}_2,\ldots,\mathbf{e}_8).
\]

In the next Section we will use the 
following Lemma.

\begin{lemma}The basis elements $\mathbf{e}_i$, $ i=1,2,3,$ satisfy the product 
relations, \label{prodrelations}
\begin{align*}
& \mathbf{e}_i\mathbf{e}_i=\frac{2}{9}\mathbf{I}-\frac{1}{3}\mathbf{e}_i,&\\
& 
\mathbf{e}_1\mathbf{e}_2=-\frac{1}{9}\mathbf{I}+\frac{1}{3}\left(\mathbf{e}
_1+\mathbf{e}_2\right),& 
\mathbf{e}_2\mathbf{e}_1=\be_5-\frac{1}{9}\mathbf{I}+\frac{1}{3}\left(\mathbf{e}
_1+\mathbf{e}_2\right),\\
& 
\mathbf{e}_2\mathbf{e}_3=-\frac{1}{9}\mathbf{I}+\frac{1}{3}\left(\mathbf{e}
_2+\mathbf{e}_3\right),& 
\mathbf{e}_3\mathbf{e}_2=\be_6-\frac{1}{9}\mathbf{I}+\frac{1}{3}\left(\mathbf{e}
_2+\mathbf{e}_3\right),\\
& 
\mathbf{e}_3\mathbf{e}_1=-\frac{1}{9}\mathbf{I}+\frac{1}{3}\left(\mathbf{e}
_3+\mathbf{e}_1\right),& 
\mathbf{e}_1\mathbf{e}_3=\be_4-\frac{1}{9}\mathbf{I}+\frac{1}{3}\left(\mathbf{e}
_3+\mathbf{e}_1\right).\\
\end{align*}
\end{lemma}

All 
structure constants can be easily found by direct computation.  Having them we 
can determine the structure constants of the automorphic Lie algebra  
$\mathfrak{A}(\infty)$
\[
 [\be_n,\be_m]=\sum_{k=1}^8 (B_{nm}^k(J)-B_{mn}^k(J))\be_k.
\]

\section{Derivation of Darboux transformations}
\label{deriveDarbouxs}

In this section Darboux transformations are constructed  for a Lax operator $L$ 
invariant with respect to the tetrahedral reduction group $\mathbb{T}$. $L$ is 
one of a pair of operators, the compatibility of which produces an integrable 
system of partial differential equations (PDEs).
To be explicit
$$L=\partial_x+u \mathbf{e}_1+v \mathbf{e}_2+w \mathbf{e}_3$$
where $uvw=1$. The second operator $A$ is given by 
$$A\,=\,\partial_t+U \mathbf{e}_1+ V \mathbf{e}_2 + W \mathbf{e}_3 + u w \mathbf{e}_4  + v u \mathbf{e}_5 + w v \mathbf{e}_6\,, $$
where
\begin{subequations}
\begin{align}
U & = \frac{u}{3}\, \left(-u+2v+2 w + \frac{u_x}{u}+2\frac{w_x}{w}\right),\\
V & =  \frac{v}{3}\, \left(-v+2w+2 u+\frac{v_x}{v}+2\frac{u_x}{u}\right),\\
W & = \frac{w}{3}\, \left(-w+2u+2 v+\frac{w_x}{w}+2\frac{v_x}{v}\right),
\end{align}
\end{subequations}
The compatibility condition $\left[L,A \right]=0$ yields the integrable system
\begin{equation} \label{eq:sys}
u_t \, =\, \partial_x U ,\quad v_t \, =\,  \partial_x V ,\quad w_t \, =\,  \partial_x W. \end{equation} 
Note that due to the relation $uvw=1$ only two of the above equations are 
independent. Using an invertible point transformation this system can rewritten 
in the form 
\[  i\psi_t=\psi_{xx}+(\psi_x^\star)^2+i(e^{-\psi-\psi^\star}+\omega e^{-\omega 
\psi-\omega^\star \psi^\star}+\omega^* e^{-\omega^\star \psi-\omega 
\psi^\star})\psi_1^*,\quad 
\omega=e^{\frac{2\pi 
i}{3}}\, ,\] 
which coinsides with equation (u3) in \cite{mr89g:58092} where the above Lax 
pair was also presented.

We search for Darboux transformations of $L$. Recall that the corresponding Darboux matrix $M$ must satisfy
\begin{equation*}
MLM^{-1}=L_1,
\end{equation*} 
where $L_1$ has the same form as $L$ but with the ``updated'' functions $u_1$, $v_1$ and $w_1$. Equivalently
\begin{equation}\label{Dcompat}
M_x=M\left(u \mathbf{e}_1+ v \mathbf{e}_2+ w \mathbf{e}_3\right)-\left(u_1 \mathbf{e}_1+v_1 \mathbf{e}_2+w_1 \mathbf{e}_3\right)M
\end{equation}
for $M$ invertible. By inspection it is clear that
(\ref{eq:sys}) is invariant with respect to cyclic permutations of the functions 
$u$, $v$ and $w$. Hence given a solution $(u,v,w)$, one can generate a new 
solution by simple permutation. This property is expressed at the level of the 
Lax structure by the existence of the following Darboux matrix 
\begin{equation*}
M_{(0)}=\left(
\begin{array}{ccc}
 0 & 0 & 1 \\
 1 & 0 & 0 \\
 0 & 1 & 0 \\
\end{array}
\right).
\end{equation*}
A short computation demonstrates that
\begin{equation*}
M_{(0)}\mathbf{e}_1M_{(0)}^{-1}=\mathbf{e}_2, \quad M_{(0)}\mathbf{e}_2M_{(0)}^{-1}=\mathbf{e}_3, \quad M_{(0)}\mathbf{e}_3M_{(0)}^{-1}=\mathbf{e}_1.
\end{equation*}
Thus we see that application of $M_{(0)}$ results in the updated functions
\begin{equation*} 
u_1 = v, \quad  v_1 = w, \quad w_1 = u. 
\end{equation*}
Whereas applying $M_{(0)}^2$ results in the updated functions
\begin{equation*} 
u_1 = w, \quad  v_1 = u, \quad w_1 = v. 
\end{equation*}

Let us construct the simplest $\lambda$--dependent  Darboux transformations.  
As a simplifying assumption it is natural to consider the matrix $M$ to be 
invariant with respect to the reduction group. The simplest such $M$ will have 
simple poles belonging to a degenerate orbit. In particular we choose these 
poles to belong to the orbit $\bbbt(\infty)$. Hence we represent $M$ in the form
\begin{equation}\label{ansatz}
M=f \mathbf{I}+a \mathbf{e}_1+b \mathbf{e}_2+c \mathbf{e}_3.
\end{equation}

\begin{theorem}
\label{dsimplepoles}Suppose $M$ as given in (\ref{ansatz}) is a Darboux matrix for the operator $L=\partial_x+u \mathbf{e}_1+v \mathbf{e}_2+w \mathbf{e}_3$, where $uvw=1$, then
\begin{equation}\label{e456eqns}
 a v_1- b u=0, \enspace{}  b w_1- c v=0,\enspace{}  c u_1- a w=0,
\end{equation}
\begin{subequations}\label{e123eqns}
\begin{align}
& a_x+\frac{1}{3}\left((b+c-a+3f)(u_1-u)+a(v_1+w_1-v-w)\right)=0\\
&  b_x+\frac{1}{3}\left((c+a-b+3f)(v_1-v)+b(w_1+u_1-w-u)\right)=0\\
&  c_x+\frac{1}{3}\left( (a+b-c+3 f)(w_1-w)+c (u_1+v_1-u-v)\right)=0.
\end{align}
\end{subequations}
and
\begin{equation}
\label{ieqn} f_x+\frac{1}{9}\left[ (2 a-b-c)(u_1-u)+ (2 b-a-c)(v_1-v)+ (2 c-a-b)(w_1-w)\right]=0
\end{equation}
\end{theorem}
\begin{proof}

The compatibility condition implies that $M$ satisfies the equation
\begin{equation*}
M_x=M\left(u \mathbf{e}_1+v \mathbf{e}_2+w \mathbf{e}_3\right)-\left(u_1 \mathbf{e}_1+ v_1 \mathbf{e}_2+ w_1 \mathbf{e}_3\right)M.
\end{equation*} 
Upon substitution of the ansatz $(\ref{ansatz})$ and an application of Lemma 
\ref{prodrelations}, we arrive at the stated equations.
From the coefficients of $\mathbf{e}_4$, $\mathbf{e}_5$ and $\mathbf{e}_6$ we obtain the system of equations (\ref{e456eqns}).
From the coefficients of  $\mathbf{e}_1$, $\mathbf{e}_2$ and $\mathbf{e}_3$ we have the equations (\ref{e123eqns}). Finally the coefficient of $\textbf{I}$ yields (\ref{ieqn}).
\end{proof}

In what follows we consider reductions of the above general system. In particular we will make use of first integrals of the system to reduce it to two components. Towards this aim we make use of the following lemma.
\begin{lemma} \label{lemmaDet}
For $M$ as given in (\ref{ansatz}), the determinant of $M$ is constant with respect to $x$ and is an automorphic function of $\lambda$.
\end{lemma} 
\begin{proof}
For any invertible matrix $M$ the following identity holds due to Jacobi's formula
\begin{displaymath}\partial_x \ln \det (M)\,=\, \text{Tr} \left(M^{-1}M_x \right).\end{displaymath}
Moreover as $M$ is a Darboux matrix for $L$ it follows that
\begin{equation*}
M^{-1}M_x=\left(u \textbf{e}_1+ v \textbf{e}_2+ w \textbf{e}_3\right)-M^{-1}\left(u_1 \textbf{e}_1+v_1 \textbf{e}_2+w_1 \textbf{e}_3\right)M.
\end{equation*}
The right hand side of the above equation consists of the difference of traceless matrices and so is necessarily traceless, hence 
combining the above two relations we conclude that
\begin{equation*}
\partial_x\det (M)=0.
\end{equation*}
The fact that $\det (M)$ is an automorphic function of $\lambda$ is a consequence of its form. By construction $M$ is invariant with respect to $\mathbb{T}$. That is for all $g\in\mathbb{T}$
\begin{equation*}
Q_gM(\sigma_g(\lambda ))Q_g^{-1}=M(\lambda )
\end{equation*}
where $Q_g$ and $\sigma_g$ are the automorphisms of $\mathfrak{sl}_3(\mathbb{C})$ and the Riemann sphere respectively, corresponding to the group element $g$. Taking the determinant of the above relation provides the equality
\begin{equation*}
\det(M(\sigma_g(\lambda )))=\det(M(\lambda )).\end{equation*}
This simply states that $\det(M(\lambda))$ is invariant with respect to the M\"obius transformation corresponding to $g$. As $g$ was an arbitrary element of $\mathbb{T}$ it follows that $\det(M)$ is an automorphic function of $\lambda$.
\end{proof}

In light of the above lemma, we express the determinant of $M$ in terms of an 
automorphic function. This will allow a number of the unknown functions $f$, 
$a$, $b$ and $c$ to be expressed in terms of the remaining ones. The 
Darboux matrix $M$ is constructed using the automorphic Lie algebra generators 
$\mathbf{e}_1$, $\mathbf{e}_2$ and $\mathbf{e}_3$, so it has poles at the 
$\mathbb{T}(\infty)$-orbit. This  means that we can expect $\det(M)$ to have 
poles on this particular orbit also. In what follows we will write $\det(M)$ in 
terms of the following automorphic function
\begin{equation} \label{eq:J1J2}
J(\lambda) \,=\,\frac{\lambda ^3 \left(\lambda ^3+8\right)^3}{4 \left(\lambda ^3-1\right)^3}.\,
%,\quad  \enspace J_2(\lambda) \,=\, J_1(\lambda) - J_1(1+\sqrt{3}). 
% \,=\,\frac{\left(\lambda ^6-20 \lambda ^3-8\right)^2}{4 \left(\lambda ^3-1\right)^3}.
\end{equation}
$J(\lambda)$ has the correct pole structure and has zeros at the orbit of zero. Using this function we arrive at the following lemma.

%We note that these are not generic. Although the structure of $M$ necessitates the use of automorphic functions with poles on %the orbit of infinity, we have some freedom in choosing the zeros of such a function. To see this, note that if
%$J(\lambda)$ is an automorphic function, then so is $J(\lambda)-J(\alpha)$, where the latter has zeros on the orbit of $\alpha$. %The functions $J_1$ and $J_2$ have zeros on the orbit of $0$ and $1+\sqrt{3}$ respectively. These are degenerated orbits %corresponding to fixed points of the mobius transformations that are involved in generating $\mathbb{T}$. 

\begin{lemma}
The determinant of $M$, as given in (\ref{ansatz}), can be expressed in the form
\begin{equation*}
\det{M}=F_1J(\lambda)+F_2\,,
\end{equation*}
for $F_1$ and $F_2$ constant with respect to $x$ and given by
\begin{subequations}
\begin{eqnarray*}
F_1 & =& \frac{abc}{16}\,,\label{F1}\\
F_2 & =& \frac{1}{27} (3 f-2 a+b+c) (3 f+a-2 b+c) (3 f+a+b-2 c). \label{F2}
\end{eqnarray*}
\end{subequations}
\end{lemma}
\begin{proof}
Result follows by direct computation and the use of Lemma \ref{lemmaDet}.
\end{proof}
As $F_1$ and $F_2$ are constants, they can be used to determine two of the unknown functions $f$, $a$, $b$ and $c$ in terms of the remaining functions. Towards this aim we may view $F_1$ and $F_2$ as elements of the ring $\mathbb{C}[a,b,c,f]$, then determining $(a,b,c,f)$ is equivalent to finding elements of the algebraic variety $V(F_1-\alpha,F_2-\beta)=V(F_1-\alpha)\cap V(F_2-\beta)$ where $\alpha$ and $\beta$ are constants in $\mathbb{C}$. The variety $V(F_1-\alpha)$ is irreducible for $\alpha\neq 0$, the case $\alpha=0$ implies at least one of $a$, $b$ or $c$ is zero, which from (\ref{e456eqns}) implies $a=b=c=0$, hence leading to trivial $M$. Thus we require $\alpha\neq 0$. For simplicity we may take $\alpha=\frac{1}{16}$, then
\begin{equation*}
\label{eq:use-F1}
c=\frac{1}{ab}.
\end{equation*} 
For general $\beta$ one finds that $V(F_2-\beta)$ is irreducible. Hence to determine $f$ in terms of $a,b$ and $c$ one would have to solve a cubic equation. However it can be seen directly that for $\beta=0$ the variety $V(F_2-\beta)$ becomes the union of three irreducible components
\begin{equation}\label{j1fact}
V(F_2)=V(3 f-2 a+b+c)\cup V(3 f+a-2 b+c)\cup V(3 f+a+b-2 c).
\end{equation}
Hence three distinct Darboux transformations can be obtained by choosing $f=f_1,f_2,f_3$ such that $(a,b,c,f_i)\in V_i$ where the $V_i$ are the irreducible components given above. As the components are linear, $f$ is expressed rationally in terms of $a$, $b$ and $c$. 

For $\beta=-abc=-1$ one finds that
\begin{equation}\label{j2fact}
V(F_2-\beta)=V(3f+a+b+c)\cup V(9f^2-3f(a+b+c)-2(a^2+b^2+c^2)+5(a b+b c+ c a))
\end{equation}
By choosing $f$ such that $(a,b,c,f)\in V(3f+a+b+c)$ we obtain a further way to express $f$ rationally in terms of $a$, $b$ and $c$.
\begin{remark}Expressing $\det(M)$ in terms of $J_2$, an automorphic function which has poles in the orbit of infinity and zeros in the orbit of $1+\sqrt 3$
\begin{equation*}
J_2=J(\lambda)-J(1+\sqrt 3)=\frac{\left(\lambda ^6-20 \lambda ^3-8\right)^2}{4 \left(\lambda ^3-1\right)^3}
\end{equation*}
one finds that $\det(M)=(F_1+abc)+F_2J_2$. Setting $(F_1+abc)=0$ leads to (\ref{j2fact}).
\end{remark}
Focusing on the irreducible components in (\ref{j1fact}), making each choice of $f$ gives the following forms for $M$ 
\begin{align*}
M_{(1)} & =\frac{2a-b-c}{3}\,\textbf{I}+a\textbf{e}_1+b\textbf{e}_2+c \textbf{e}_3\\
M_{(2)} & =\frac{2b-a-c}{3}\,\textbf{I}+a\textbf{e}_1+b\textbf{e}_2+c \textbf{e}_3\\
M_{(3)} & =\frac{2 c-a-b}{3}\,\textbf{I}+a\textbf{e}_1+b\textbf{e}_2+c \textbf{e}_3
\end{align*}
where $c=1/ab$ and $a$ and $b$ are determined as in Theorem \ref{dsimplepoles}. However it becomes clear from (\ref{F2}) that the permutation 
\begin{equation*} \label{eq:per-mnl}
\pi : (a,b,c) \mapsto (b,c,a), \quad \pi^3 = {\rm{id}}
\end{equation*}
 maps one factor of $F_2$ into another. Letting $T$ denote the Lie algebra automorphism,
\begin{equation*}
T: A\mapsto M_{(0)}AM_{(0)}^{-1}
\end{equation*}
it is easy to see that
\begin{equation}\label{Drelations}
M_{(2)}=\pi\left(T(M_{(1)})\right),\enspace M_{(3)}=\pi\left(T(M_{(2)})\right), \enspace M_{(1)}=\pi\left(T(M_{(3)})\right).
\end{equation}
Thus $M_{(1)}$, $M_{(2)}$ and $M_{(3)}$ are related by conjugation by $M_{(0)}$ and a cyclic relabelling of the functions $a$, $b$ and $c$. 

Turning attention to the irreducible components in (\ref{j2fact}) we see that each component is invariant with respect to $\pi$, so in particular denoting 
\begin{equation}
M_{(4)}\,=\,-\, \frac{a+b+c}{3}\,\textbf{I}+a\textbf{e}_1+b\textbf{e}_2+c \textbf{e}_3
\end{equation}
it follows that
\begin{equation*}
\pi\left(T(M_{(4)})\right)=M_{(4)}.
\end{equation*}

By expressing $f$ and $c$ in terms of the remaining functions using $F_1$ and $F_2$ as described above, their presence can be removed from Theorem \ref{dsimplepoles}. First integrals were used to determine $f$ and $c$, therefore the reduction is compatible with the original system. In order to separate the different cases, we use different indices for the updated functions corresponding to each Darboux transformation. In particular, $u_i$ and $v_i$ denote the updated functions corresponding to $M_{(i)}$, $i=1,\cdots,4$. This separation is also convenient for later use when we discuss fully discrete systems. 

\begin{theorem}\label{nonlocalsyms} $M_{(i)}$ for $i=1, \ldots ,4,$ is a Darboux transformation for $L$ provided
\begin{equation}a = c u u_i v,\enspace b= c u_i v v_i,\enspace c^3 u u_i^2 v^2 v_i=16\end{equation}
and the following relations hold:

\begin{subequations}\label{nonLocalSyms1}
\begin{align}
& \frac{v_{1,x}}{v_1}-\frac{u_{x}}{u}-\frac{u v}{v_1}+u+v-v_1=0\\
& \frac{v_{x}}{v}+\frac{u_{1,x}}{u_{1}}+\frac{u_{x}}{u}-\frac{1}{u v}-\frac{u v}{v_1}+\frac{1}{u_1 v_1}+u_1=0
\end{align}
\end{subequations}
\begin{subequations}\label{nonLocalSyms2}
\begin{align}
& \frac{v_{2,x}}{v_2}-\frac{u_{x}}{u}-\frac{u_2 v_2}{u}-u+u_2+v_2=0\\
& \frac{v_{x}}{v}+\frac{u_{2,x}}{u_{2}}+\frac{u_{x}}{u}-\frac{u_2 v_2}{u}-u+u_2+v_2=0
\end{align}
\end{subequations}
\begin{subequations}\label{nonLocalSyms3}
\begin{align}
&  \frac{v_{3,x}}{v_3}-\frac{u_{x}}{u}-\frac{1}{u v}+\frac{2}{u_3 v}-\frac{1}{u_3 v_3}-u+u_3+v-v_3=0\\
& \frac{v_{x}}{v}+\frac{u_{3,x}}{u_{3}}+\frac{u_{x}}{u}+\frac{1}{u v}-\frac{1}{u_3 v}+u-u_3=0
\end{align}
\end{subequations}
\begin{subequations}\label{nonLocalSyms4}
\begin{align}
& \frac{v_{4,x}}{v_4}-\frac{u_{x}}{u}-u+u_4+v-v_4=0\\
&  \frac{v_{x}}{v}+\frac{u_{4,x}}{u_{4}}+\frac{u_{x}}{u}-\frac{1}{u v}+\frac{1}{u_4 v_4}+u-u_4=0
\end{align}
\end{subequations}
\end{theorem}
\begin{proof}
The theorem follows from Theorem \ref{dsimplepoles} by making the previously described reductions.
\end{proof}

\begin{remark}\label{pointTransRemark}
As previously stated, the above Darboux transformations are related via (\ref{Drelations}), this relation can be seen at the level of the systems above. In particular the systems corresponding to $M_{(1)}$ and $M_{(2)}$ are related by the point transformation 

\begin{equation}\label{pointTransUV}
(u,v)\mapsto (v,u^{-1}v^{-1})
\end{equation}
and a relabelling of index. The same is true for the pairs of systems corresponding to $M_{(2)}$ and $M_{(3)}$, also for $M_{(3)}$ and $M_{(1)}$.  
\end{remark}

%(Expressions left here to be copy and pasted elsewhere)
%\begin{align}& M_{(1)}=c\left(\frac{2 u u_1 v-u_1 v v_1 -1}{3}\,\textbf{I}+u u_1 v \textbf{e}_1+u_1 v v_1\textbf{e}_2+ %\textbf{e}_3\right)\\
%&  M_{(2)}=c\left(\frac{2u_2 v v_2- u u_2 v -1}{3}\,\textbf{I}+u u_2 v \textbf{e}_1+u_2 v v_2\textbf{e}_2+ %\textbf{e}_3\right)\\
%&  M_{(3)}=c\left(\frac{2 u u_3 v-u_3 v v_3 -1}{3}\,\textbf{I}+u u_3 v \textbf{e}_1+u_3 v v_3\textbf{e}_2+ %\textbf{e}_3\right)\\
%& M_{(4)}=c\left(-\,\frac{u_4 u v+ u_4 v_4 v+1}{3}\textbf{I}+u u_4 v \textbf{e}_1+u_4 v v_4 \textbf{e}_2+ %\textbf{e}_3\right),
%\end{align}

\section{Darboux Lax representation for local symmetries}
\label{localSymDerivation}
The derived Darboux matrices may be regarded as defining shifts on a lattice. To make it clear, we may impose that the functions $u$ and $v$ depend on a set of discrete variables $k_i$, $i=1,\ldots,4$, and interpret the updated functions $u_i$ and $v_i$ as shifts in the corresponding lattice direction. In this interpretation, the Lax operator and its Darboux transformations define Darboux-Lax representations (semi-discrete Lax pairs) of the corresponding systems in Theorem (\ref{nonlocalsyms}). These systems are then viewed as differential difference equations. Here and in what follows, let ${\cal{S}}_i$ denote the shift operator in the $i$-th direction, i.e. ${\cal{S}}^j_i(f) = f(k_i+j)$. Then explicitly, the differential difference relations corresponding to $M_{(i)}$ are equivalent to the compatibility of the following system
$${\cal{S}}_i (\Psi)  = M_{(i)} \Psi,\quad \Psi_x = - \left(u \mathbf{e}_1+v \mathbf{e}_2+w \mathbf{e}_3 \right) \Psi. $$
From this viewpoint equations in (\ref{nonlocalsyms}) are integrable. However they are not evolutionary.

We now turn our attention to deriving evolutionary differential difference equations associated with each $M_{(i)}$. The problem then becomes to find continuous flows  $\Psi_{\tau^i}= \Omega_{(i)} \Psi$  such that the compatibility condition  
\begin{equation}\label{compatEvol}
\partial_{\tau^i} {M_{(i)}}\,=\,{\cal{S}}_i \left(\Omega_{(i)}\right) M_{(i)} - M_{(i)} \Omega_{(i)}\,,
\end{equation}
is equivalent to an evolutionary system of differential difference equations. We tackle this problem by taking an ansatz for $\Omega_{(i)}$ of the following form

\begin{equation} \label{eq:Om}
\Omega_{(i)} = M_{(i)}^{-1}(g {\rm{I}} + p  \mathbf{e}_1  + q  \mathbf{e}_2  + r  \mathbf{e}_3).
\end{equation}
This ansatz is based upon the observations  that sometimes the Lax operator 
corresponding to a local symmetry can be written in this form \cite{X}. To 
reduce the number of unknown functions in $\Omega_{(i)}$ we make use of the 
following lemma.

\begin{lemma}\label{traceConstant}
We may assume $\Omega_{(i)}$ has constant trace with respect to $\tau^i$.
\end{lemma}
\begin{proof}
Due to the choice of $F_1$ and $F_2$ taken when deriving each $M_{(i)}$, it follows that $\det{(M_{(i)})}$ is constant with respect to $\tau_i$. This fact along with Jacobi's identity yields $\left( {\cal{S}}_i-{\rm{I}}\right) \left({\rm{Tr}}(\Omega_{(i)})\right) = 0$. Hence ${\rm{Tr}}(\Omega_i)$ is constant with respect to ${\cal{S}}_i$. Now consider the transformation
$$\tilde{\Omega}_{(i)}\mapsto \Omega_{(i)}-\frac{({\rm{Tr}}(\Omega_i)-\alpha)}{3}{\rm{I}}$$
where $\alpha$ is constant with respect to $\tau^i$ and the shift operator ${\cal{S}}_i$. If $\Omega_{(i)}$ satisfies (\ref{compatEvol}), then as ${\rm{Tr}}(\Omega_{(i)})$ is constant with respect to ${\cal{S}}_i$ it follows that $\tilde{\Omega}_i$ also satisfies (\ref{compatEvol}). Hence we may replace $\Omega_{(i)}$ by $\tilde{\Omega}_{(i)}$ to obtain an operator with constant trace. Moreover for constant $k$ 
\begin{align*}
\Omega_{(i)}-k{\rm{I}} & = M_{(i)}^{-1}(g {\rm{I}} + p \mathbf{e}_1 + q\mathbf{e}_2 + r \mathbf{e}_3)-k{\rm{I}}\\
& = M_{(i)}^{-1}(g {\rm{I}} + p \mathbf{e}_1 + q\mathbf{e}_2 + r \mathbf{e}_3)-k{\rm{I}}\\
& = M_{(i)}^{-1}(g {\rm{I}} + p \mathbf{e}_1 + q\mathbf{e}_2 + r \mathbf{e}_3)-kM^{-1}_{(i)}M_{(i)}\\
& = M^{-1}_{(i)}((g-f k){\rm{I}}+(p-k a)\mathbf{e}_1+(q-k b)\mathbf{e}_2+(r-k c)\mathbf{e}_2).
\end{align*}
where $a$, $b$, $c$ and $f$ are as in Theorem \ref{nonlocalsyms}. Hence by relabelling appropriately $\tilde{\Omega}_{(i)}$ is of the form (\ref{eq:Om}). 
\end{proof}

Armed with the previous lemma we proceed to solve the compatibility condition for $\Omega_{(i)}$. This leads to the following result.
 Note that negative indices denote backward shifts, i.e. $u_{-i} = u(k_i - 1)$, $v_{-i} = v(k_i-1)$.
\begin{theorem}\label{theoremLocal}$\Omega_{(1)}$ is determined by the relations

\begin{align*}
 r & = -\,\frac{p}{u u_1 v}\,-\,\frac{q}{u_1 v v_1}\,,\quad g  =\frac{u u_1 v v_1}{(u u_1 v-1) (u-v_1)}\,+\,\frac{2 p -q-r}{3}\,,\\
p & =\frac{1}{3}\,\frac{u u_1 v}{u -v_1}\left(\,\frac{2 v_1}{u u_1 v-1}- \frac{u v + (u_{-1}-2 v) v_1}{(u_{-1} u v_{-1}-1) (u_{-1}-v)}\, \right), \\
q & = \frac{1}{3}\,\frac{u_1 v v_1}{u -v_1}\left( \,\frac{ 2 u v-(u_{-1}+ v) v_1}{(u_{-1} u v_{-1}-1) (u_{-1}-v)}-\frac{v_1}{u u_1 v-1}\, \right).
\end{align*}
The corresponding differential difference system is given by
\begin{subequations}
\label{localsym1}
\begin{align} 
u_{\tau^1} & = u({\cal{S}}_1-{\rm{I}})\frac{ v}{ \left(v-u_{-1}\right) \left(u u_{-1} v_{-1}-1\right)}\,\\
v_{\tau^1} & = \frac{v}{u u_{-1} v_{-1}-1}({\cal{S}}_1-{\rm{I}})\frac{u_{-1}}{u_{-1}-v}. 
\end{align}
\end{subequations}
\end{theorem}
\begin{proof}
Lemma \ref{traceConstant} gives that ${\rm{Tr}}(\Omega_{(i)})$ is constant. Moreover from the form of $\Omega_{(i)}$ it follows that its trace is an automorphic function of $\lambda$. Bearing this in mind we find that ${\rm{Tr}}(\Omega_{(i)})=\epsilon_1\frac{1}{J(\lambda)}+\epsilon_2$ where $J(\lambda)$ is given in (\ref{eq:J1J2}), while $\epsilon_1$ and $\epsilon_2$ are as follows
\begin{align*} 
\epsilon_1 & =\frac{16 \left(u u_1 v-1\right) \left(u-v_1\right) \left(3 g-2 p+q+r\right)}{3 u u_1 v v_1}\\
\epsilon_2 & = \frac{p}{u u_1 v}+\frac{q}{u_1 v v_1}+r.
\end{align*}
The above relations allow $g$ and $r$ to determined. Moving on to the compatibility condition we find that $u$ and $v$ must satisfy the differential difference equations
\begin{align*}
 & \frac{u_{1,\tau}}{u_1}+\frac{u_{\tau}}{u}+\frac{v_{\tau}}{v}+\frac{p}{u u_1 v}+\frac{p_1}{u_1 u_2 v_1}-\frac{q_1}{u_2 v_1 v_2}-r=0\\
& \frac{v_{1,\tau}}{v_1}-\frac{u_{\tau}}{u}-\frac{2 p}{u u_1 v}+\frac{p_1}{u_1 u_2 v_1}+\frac{2 q_1}{u_2 v_1 v_2}-r=0.
\end{align*}
This system implies the following condition on $u$,
\begin{equation*}
(I+\mathcal{S}_1+\mathcal{S}_1^2)\frac{u_{\tau}}{u}=\frac{2 q_1}{u_2 v_1 v_2}+\frac{q_2}{u_3 v_2 v_3}-\frac{2 p}{u u_1 v}-\frac{p_2}{u_2 u_3 v_2}-r+r_1.
\end{equation*}
\begin{equation*}\implies
(I+\mathcal{S}_1+\mathcal{S}_1^2)\frac{u_{\tau}}{u}=(I+\mathcal{S}_1+\mathcal{S}_1^2)\left(\frac{p v_1-q u}{uu_1 v v_1}\right).
\end{equation*}
The operator $(I+\mathcal{S}_1+\mathcal{S}_1^2)$ has kernel $\omega^{k_1}$ where $\omega$ is a cube root of unity, so in general from the above relation we can determine $u_{\tau}$ up to the addition of a constant multiple of $\omega^{k_1}u$. However this corresponds to the addition of a point symmetry.
We are primarily concerned with generalised symmetries and so in the following analysis we shall omit the kernel. Inverting the operator $(I+\mathcal{S}_1+\mathcal{S}_1^2)$ we obtain the equation
\begin{equation*}
u_{\tau}=\frac{p v_1-q u}{u_1 v v_1}.
\end{equation*}
Then the corresponding equation for $v$ is obtained as
\begin{equation*}
v_{\tau}=v \epsilon _2-\frac{p}{u u_1}-\frac{2 q}{u_1 v_1}.
\end{equation*}
Now returning to the compatibility condition we can solve the remaining equations to determine $p$ and $q$. Doing so we find that
\begin{align*}
p & =\frac{u u_1 v \epsilon _1 \left(u v+u_{-1} v_1-2 v_1 v\right)}{48 \left(u_{-1} u v_{-1}-1\right) \left(u_{-1}-v\right) \left(u-v_1\right)}+\frac{u u_1 v \left(8 \epsilon _2 \left(u u_1 v-1\right) \left(u-v_1\right)-v_1\epsilon_1\right)}{24 \left(u u_1 v-1\right) \left(u-v_1\right)}\\
q & = \frac{-u_1 v v_1 \epsilon _1 \left(2 u v-u_{-1} v_1-v_1 v\right)}{48 \left(u_{-1} u v_{-1}-1\right) \left(u_{-1}-v\right) \left(u-v_1\right)}+\frac{u_1 v v_1 \left(16 \epsilon _2 \left(u u_1 v-1\right) \left(u-v_1\right)+v_1\epsilon_1 \right)}{48 \left(u u_1 v-1\right) \left(u-v_1\right)}.
\end{align*}
Finally, eliminating $p$ and $q$ from the differential difference equations above, we arrive at
\begin{equation*}
u_{\tau} =\frac{-\epsilon_1 u}{16}(S-I)\frac{ v}{ \left(v-u_{-1}\right) \left(u u_{-1} v_{-1}-1\right)}, \enspace{} v_{\tau} =\frac{-\epsilon_1 v}{16 \left(u u_{-1} v_{-1}-1\right)}(S-I)\frac{u_{-1}}{\left(u_{-1}-v\right)}.
\end{equation*}
The compatibility condition is now completely satisfied. Note that $\epsilon_1$ can be an arbitrary non zero constant, in particular take $\epsilon_1=-16$. Moreover $\epsilon_2$ can be set to zero without affecting the computed flow. Making these choices leads to the expressions in the theorem.
\end{proof}

\begin{corollary}
The differential difference equations corresponding to $\Omega_{(2)}$ and $\Omega_{(3)}$ are given by
\begin{align*}
u_{\tau^2} & = \frac{u}{\left(u_2 v v_2-1\right)}(\mathcal{S}_2-{\rm{I}})\frac{ v}{\left(v-u_{-2}\right) },\\
\quad v_{\tau^2} & = v(\mathcal{S}_2-{\rm{I}})\frac{u_{-2}}{{\left(u v_{-2} v-1\right)}\left(u_{-2}-v\right)},
\end{align*}
and
\begin{align*}
u_{\tau^3} & = \frac{u^2v_{-3}v}{u v_{-3} v-1}(\mathcal{S}_3-{\rm{I}})\frac{1}{u_{-3} u v_{-3}-1},\\
v_{\tau^3} & = \frac{u u_3 v^2}{u u_3 v-1}(\mathcal{S}_3-{\rm{I}})\frac{1}{u v_{-3} v-1},
\end{align*}
respectively.
\end{corollary}
\begin{proof}
The result follows from (\ref{localsym1}) using the transformation (\ref{pointTransUV}) and relabelling the indices.
\end{proof}

Following the same procedure as in the proof of Theorem \ref{theoremLocal}  we derive the corresponding result for $M_{(4)}$.
\begin{theorem} \label{theoremlocalsym4}
$\Omega_{(4)}$ is determined by the relations
\begin{align*}
 r & = -\,\frac{p}{u u_4 v}\,-\,\frac{q}{u_4 v v_4}\,,\quad g  =\frac{u u_4 v v_4}{u u_4 v v_4 +u +v_4}\,-\,\frac{p+q+r}{3}\,,\\
p & = \frac{u u_4 v}{3}  ({\cal{S}}_4-{\rm{I}}) \frac{u_{-4} + 2 v}{u_{-4} u v_{-4} v + u_{-4} + v}\,,\\
\quad q & = \frac{u_4 v v_4}{3}  ({\cal{S}}_4-{\rm{I}}) \frac{u_{-4} - v}{u_{-4} u v_{-4} v + u_{-4} + v}.
\end{align*}
The corresponding differential difference system is given by
\begin{subequations}\label{eq:localsym2}
\begin{align}
u_{\tau^4} & = u ({\cal{S}}_4-{\rm{I}}) \frac{v}{u_{-4} u v v_{-4}+u_{-4}+v}\\ 
v_{\tau^4} & = v({\cal{S}}_4-{\rm{I}})\frac{u_{-4}}{u_{-4} u v v_{-4}+u_{-4}+v}.
\end{align}
\end{subequations}
\end{theorem} 

\section{Darboux pairs and fully discrete systems}
\label{discretesysDerivations}
Viewing Darboux transformations as shifts on a lattice allowed differential-difference equations to be derived via a Darboux-Lax pair (semi-discrete Lax pair). We now derive fully discrete equations as the compatibility condition of two separate Darboux transformations. The transformations act as a Darboux pair, more commonly referred to as a fully discrete Lax pair. For the system
\begin{equation*} \label{eq:fdLP}
{\cal{S}}_i \left(\Psi\right)\,=\, M_{(i)} \Psi\,,\quad {\cal{S}}_j \left(\Psi\right)\,=\, M_{(j)} \Psi\,,\quad i\ne j,
\end{equation*}
it follows that the compatibility condition is given by
\begin{equation*} \label{eq:compcond}
{\cal{S}}_j(M_{(i)}) M_{(j)}\,-\,{\cal{S}}_i(M_{(j)}) M_{(i)}\,=\,0\,.
\end{equation*}
We shall see that the last relation yields a system of partial difference equations which we denote as $Q_{(i,j)}$. So in fact the consistency of the Darboux pair is equivalent to a system of partial difference equations. Requiring that the Darboux transformations derived in the previous section are compatible leads to the following theorem. We note that here we have considered only systems $Q_{(1,2)}$ and $Q_{(1,4)}$. This is because we may employ the point transformations described in (\ref{pointTransRemark}) to map system $Q_{(1,2)}$ to systems $Q_{(2,3)}$ and $Q_{(3,1)}$, respectively. In the same fashion, system $Q_{(1,4)}$ can be mapped to systems $Q_{(2,4)}$ and $Q_{(3,4)}$, respectively. 

\begin{theorem} \label{theoremDIS}
The Darboux pairs $(M_{(1)},M_{(2)})$ and $(M_{(1)},M_{(4)})$ are compatible provided the following systems hold.
\begin{subequations}
\label{eq:Q12}
\begin{align} 
&  u_{2} v_{2} \left[u_{1} v \left(v_{1}-u\right)+u_{12} v_{1} \left(u_{1}-v_{12}\right)\right]+u_{2} v_{2}-u_{1} v_{1}=0,\\
& u v_{12}-u_2 v_2=0.
\end{align}
\end{subequations}
\begin{subequations}
 \label{eq:Q14}
\begin{align}
& u v_{14} -u_4 v_4+v_{14} v_4-v_1 v_{14}=0,\\
& \frac{u v}{v_1}-\frac{1}{u_1 v_1}+\frac{1}{u_4 v_4}-u_{14}=0.
\end{align}
\end{subequations}
It should be noted that in the above systems, double indices denote shifts by one step in both lattice directions, i.e. $u_{ij} = u(k_i+1,k_j+1)$ and $v_{ij} = v(k_i+1,k_j+1)$.
\end{theorem}
\begin{proof}
The theorem follows from direct computation.
\end{proof}

To the best of our knowledge, the above systems are new. They are integrable as they follow from the compatibility condition of a discrete Lax pair and, as we explain below, they admit symmetries in both lattice directions. Another manifestation of their integrability is their multi-dimensional consistency. If we consider any three systems  $Q_{(i,j)}$, $Q_{(j,k)}$ and $Q_{(k,i)}$, where $i\ne j \ne k \ne i$, then we can compute $u_{ijk}$ and $v_{ijk}$ in three different ways, provided by the corresponding systems, all of which yield the same result. Another interpretation of this construction is that two of the systems define an auto-B{\"a}cklund transformation of the remaining one. 

\begin{proposition}
Suppose $M$, $N$ where $M\Psi=\Psi_i$ and $N\Psi=\Psi_j$ are Darboux transformations for the operator $L=\partial_{\tau}-A$, then 
\begin{equation}
\frac{d}{d\tau}(M_jN-N_iM)=0
\end{equation}
when restricted to  $M_jN-N_iM=0$. 
\end{proposition}

\begin{proof}
As $M$, $N$ are Darboux transformations for $L$ it follows that
\begin{displaymath}M_{\tau}=A_iM-MA\end{displaymath}
\begin{displaymath}N_{\tau}=A_jN-NA.\end{displaymath}
Hence 
\begin{align*}\frac{d}{d\tau}(M_jN-N_iM) & = (M_j)_{\tau}N+M_jN_{\tau}-(N_i)_{\tau}M-N_iM_{\tau} \\
& = U_{ji}(M_jN-N_iM)-(M_jN-N_iM)U\\
& = 0.
\end{align*}
\end{proof}
\begin{corollary}
The differential-difference equations derived in (\ref{nonlocalsyms}) define non-local symmetries of the corresponding discrete systems $Q_{(i,j)}$. 
\end{corollary}
One can confirm directly that the $\tau^i$-flow and $\tau^j$-flow define local generalised symmetries in the $i$-th and $j$-th lattice direction respectively, for system $Q_{(i,j)}$.

\section{Related Integrable Systems}
\label{relatedIntSystems}
In this section systems related to $Q_{(1,2)}$ and $Q_{(1,4)}$ are presented.
For system $Q_{(1,2)}$ we have the following related six point equation and local symmetries.
\begin{theorem} \label{theorempot}
The system $Q_{(1,2)}$, as given in (\ref{eq:Q12}), is related via the potentiation
\begin{equation}u=\frac{\phi}{\phi_{1}},\enspace{} v=\frac{\phi_{-2}}{\phi},\end{equation} 
to the six-point equation
\begin{equation}\label{poteqn}
\left(\phi_1-\phi_2\right)\phi \phi_{1122}+\left(\phi_2\phi_{112}-\phi_1\phi_{122}\right)\phi_{1122}+\left(\phi_{122}-\phi_{211}\right)\phi_1\phi_2=0.
\end{equation}
With corresponding symmetries
\begin{equation}
\phi_{\tau^1}\,=\,\frac{\phi \phi_{1} \phi_{-2}}{(\phi_{-1}-\phi_{-2}) (\phi_1 - \phi_{-1-2})}\,,\quad \phi_{\tau^2}\,=\,\frac{\phi \phi_{2} \phi_{-1}}{(\phi_{-2}-\phi_{-1}) (\phi_2 - \phi_{-1-2})}\,.
\end{equation} 
\end{theorem}
\begin{proof}
We observe that the second equation in system $Q_{(1,2)}$ has the form of a conservation law as it can be written as
$$ \left( {\cal{S}}_1-{\rm{I}}\right) \log v_2 \,=\,  \left( {\cal{S}}_2-{\rm{I}}\right) \log u .$$
This suggests the introduction of a potential $\phi$ via the relations $v_2 = \phi/\phi_2$ and $u = \phi/\phi_1$. In terms of $\phi$, the second equation of $Q_{(1,2)}$ becomes an identity, whereas the first equation becomes the six-point scalar equation in the statement of the theorem. The $\tau^1$ symmetry follows from the introduction of $\phi$ into (\ref{localsym1}). The $\tau^2$ symmetry follows from the invariance of (\ref{poteqn}) under the interchange of directions 1 and 2.
\end{proof} 
Here, repeated indices denote higher order shifts in the corresponding direction, e.g. $\phi_{1122} = \phi(k_1+2,k_2+2)$. Considering the system $Q_{(1,4)}$ we arrive at the following result.

\begin{theorem} \label{theoremreduction}
The system $Q_{(1,4)}$, as given in (\ref{eq:Q14}), is related to the six-point equation
\begin{equation}\label{eq:redQ14} 
\frac{1}{\psi_{114}}\,+\,\psi_1 \psi_{11}\,=\, \psi_{4} \psi_{14}\,+\,\frac{1}{\psi}\,
\end{equation}
under the reduction $v_1=u=\frac{1}{\psi}$. Local symmetries of (\ref{eq:redQ14}) are given by
\begin{subequations}\label{symsred}
\begin{align}
\psi_{\tau^4} & = \psi\left({\rm{I}}- {\cal{S}}_4^{-1}\right)\frac{\psi \psi_1 \psi_{14}}{\psi \psi_{14} (\psi_1 + \psi_{4})+1}\\
\psi_{s^1} & = \psi\left({\rm{I}}-{\cal{S}}_1^{-1}\right)\frac{1}{(\psi_{11} \psi_1 \psi -1) (\psi_1 \psi \psi_{-1} -1) }.
\end{align}
\end{subequations}
\end{theorem}
\begin{proof}
Focusing on the first equation of $Q_{(1,4)}$, it can be easily verified that it turns into an identity under the reduction $u = 1/\psi_1$, $v = 1/\psi$ (equivalently $u=v_1$), whereas the second equation of the same system becomes the six-point equation (\ref{eq:redQ14}). This reduction applied to (\ref{eq:localsym2}) produces the $\tau^4$ flow given above. The symmetry corresponding the $\tau^1$ direction is not compatible with this reduction. To obtain the $s^1$ flow presented in (\ref{symsred}) it is profitable to note that the six-point equation may be rewritten in the following form
\begin{equation}\label{eq:redQ14a} 
\big(\psi {\cal{S}}_1-\psi_{114} \big) \left(\psi \psi_{14} (\psi_1+\psi_{4})+1\right)\,=\,0.
\end{equation}
The quadrilateral equation
\begin{equation} \label{eq:MX}
\psi \psi_{14} (\psi_1+\psi_{4})+1 = 0
\end{equation}
defines particular solutions of equation (\ref{eq:redQ14a}) and was derived in \cite{MX} in the context of second order integrability conditions for difference equations. Knowing from \cite{MX} that equation (\ref{eq:MX}) admits a second order symmetry in both directions, we can easily show that 
\begin{equation*}\label{miuraSym2}
\psi_{s^1}\,=\psi\left({\rm{I}}-{\cal{S}}_1^{-1}\right)\frac{1}{(\psi_{11} \psi_1 \psi -1) (\psi_1 \psi \psi_{-1} -1) }
\end{equation*}
also defines a symmetry for the six-point equation (\ref{eq:redQ14}) in the first lattice direction. It is worth noting that the denominator of the $\tau^4$ symmetry is the defining polynomial of equation (\ref{eq:MX}). \end{proof}

Finally we have the following result regarding Miura transformations of the 6-point equation (\ref{eq:redQ14}).
\begin{theorem} \label{theoremmiura}
The system (\ref{eq:redQ14}) and its symmetries (\ref{symsred})
are linked via the Miura transformations
\begin{equation*}
\chi\,=\, \frac{1}{1-\psi \psi_1 \psi_{11}}\,,\enspace{} \rho\,=\, \frac{\psi \psi_4 \psi_{14}}{\psi \psi_{14} (\psi_1 + \psi_{4})+1}\,, 
\end{equation*}
to the systems
\begin{equation} \label{eq:eq-chi}
\chi_1 \chi_{11} (\chi-1) \,=\, \chi_4 \chi_{14} (\chi_{114}-1)\,,
\end{equation}
\begin{align*}
& \chi_{\tau^4}\,=\,\chi (\chi-1) \left( {\rm{I}}- {\cal{S}}_4^{-1}\right) \frac{\chi_1}{\chi_4 (\chi_{14} + \chi_1-1) + \chi_1 (\chi-1)},\\\
& \chi_{s^1}\,=\, \chi (\chi-1) \left(\chi_{11} \chi_1 - \chi_{-1} \chi_{-1-1}\right),
\end{align*}
and
\begin{equation}\label{eq:eq-r}
\rho_{114} (\rho_1 + \rho_{11}-1) \,=\, \rho (\rho_4 + \rho_{14}-1),
\end{equation}
\begin{align*}
& \rho_{\tau^4}\,=\,\rho \,\left( {\rm{I}}-{\cal{S}}_4^{-1}\right) \big(\rho_{14} (1-\rho - \rho_{1}) + \rho + \rho_{4} - \rho \rho_{4} \big),\\
& \rho_{s^1}\,=\,\rho\,\frac{ (\rho_1+\rho-1) (\rho + \rho_{-1}-1)}{\rho_1 + \rho + \rho_{-1}-1} \,\left({\rm{I}}- {\cal{S}}_1^{-2}\right) \frac{1}{\rho_{11} + \rho_1 + \rho-1},
\end{align*} 
respectively.
\end{theorem}
\begin{proof} Can be directly verified.\end{proof}

\begin{remark}
{\rm{The six-point equations (\ref{eq:redQ14}), (\ref{eq:eq-chi}) and (\ref{eq:eq-r}) may be considered as discrete analogues of third order hyperbolic equations \cite{ASh}. To the best of our knowledge, equations  (\ref{eq:redQ14}) and (\ref{eq:eq-r}) are new, and equation (\ref{eq:eq-chi}) appeared for the first time in a different context in \cite{AP1}.}}
\end{remark}

\section{Conclusions}

In this paper we presented the Lax-Darboux scheme related to the tetrahedral reduction group. More precisely, we studied the Darboux transformations for the class of Lax operators which involve $3 \times 3$ matrices and are invariant under the action of this group. The discrete character of these transformations allowed us to derive various discrete  systems, all of which are integrable. Specifically, we derived new systems of differential-difference equations which we presented in Theorems \ref{theoremLocal} and \ref{theoremlocalsym4}. These semi-discrete equations define generalised symmetries of the new systems of difference equations  given in Theorem \ref{theoremDIS}. The latter systems are also related to certain scalar difference equations by means of potentiation and reductions, as they are desribed in Theorems \ref{theorempot} and \ref{theoremreduction}, respectively. Finally, we presented two different Miura transformations in Theorem \ref{theoremmiura} which relate system (\ref{eq:redQ14}) and its symmetries to equations (\ref{eq:eq-chi}) and (\ref{eq:eq-r}) and their corresponding symmetries, respectively.

\section*{Acknowledgements}
This work was supported  by the EPSRC grant EP/I038675/1 (AVM and PX) and by a grant from the Leverhulme Trust  (GB and AVM)


\begin{thebibliography}{10}

\bibitem{mik79}
A.V. Mikhailov.
\newblock Integrability of a two-dimensional generalization of the {T}oda
chain.
\newblock {\em JETP Lett.}, 30(7):414--418, 1979.

\bibitem{mik80}
A.V. Mikhailov.
\newblock Reduction in integrable systems. {T}he reduction group.
\newblock {\em JETP Lett.}, 32(2):187--192, 1980.

\bibitem{mik81}
A.V. Mikhailov.
\newblock The reduction problem and the inverse scattering method.
\newblock {\em Phys. D}, 3(1\& 2):73--117, 1981.

\bibitem{mik_dis}
Mikhailov~A. V.
\newblock Reduction method in the theory of integrable equations and its
applications to problems of magnetism and nonlinear optics.
\newblock 1987.
\newblock Sc.D. Thesis, L.D.Landau Institute for Theoretical Physics, (in
Russian).

\bibitem{LM05}
S.~Lombardo and A.V. Mikhailov.
\newblock Reduction groups and automorphic {L}ie algebras.
\newblock {\em Communications in Mathematical Physics}, 258:179--202, 2005.

\bibitem{LM04}
S.~Lombardo and A.V. Mikhailov.
\newblock Reductions of integrable equations: dihedral group.
\newblock {\em Journal of Physics A: Mathematical and General}, 37:7727--7742,
2004.

\bibitem{sara}
S.~Lombardo.
\newblock {\em Reductions of Integrable Equations and Automorphic {L}ie
	Algebra}.
\newblock PhD thesis, University of Leeds, Leeds, 2004.

\bibitem{bm2}
R.~Bury and A.V. Mikhailov.
\newblock Automorphic {L}ie algebras and corresponding integrable systems.
\newblock 2009.
\newblock Draft.

\bibitem{bury}
R.T. Bury.
\newblock {\em Automorphic Lie Algebras, Corresponding Integrable Systems and
	their Soliton Solutions}.
\newblock PhD thesis, University of Leeds, Leeds, 2010.

\bibitem{ls10}
Sara Lombardo and Jan~A. Sanders.
\newblock On the classification of automorphic {L}ie algebras.
\newblock {\em Communications in Mathematical Physics}, 299(3):793--824, 2010.

\bibitem{kls14}
V~Knibbeler, S~Lombardo, and J~A Sanders.
\newblock Automorphic {L}ie algebras with dihedral symmetry.
\newblock {\em Journal of Physics A: Mathematical and Theoretical},
47(36):365201, 2014.

\bibitem{kls15}
V.~{Knibbeler}, S.~{Lombardo}, and J.~A. {Sanders}.
\newblock {Higher dimensional Automorphic Lie Algebras}.
\newblock {\em Found Comput Math}, doi:10.1007/s10208-016-9312-1, 2016.

\bibitem{smx}
S.~Konstantinou-Rizos, A.~V. Mikhailov, and P.~Xenitidis.
\newblock Reduction groups and related integrable difference systems of
nonlinear {S}chr{\"o}dinger type.
\newblock {\em Journal of Mathematical Physics}, 56(8):082701, 2015.

\bibitem{mr89g:58092}
A.~V. Mikhailov, A.~B. Shabat, and R.~I. Yamilov.
\newblock Extension of the module of invertible transformations.
{C}lassification of integrable systems.
\newblock {\em Comm. Math. Phys.}, 115(1):1--19, 1988.

\bibitem{BL2002}
V.~M. Buchstaber and D.~V. Leykin.
\newblock Polynomial {L}ie {A}lgebras.
\newblock {\em Functional Analysis and Its Applications}, 36(4):267--280, 2002.

\bibitem{X}
P.~Xenitidis.
\newblock Integrability and symmetries of difference equations: the
{A}dler--{B}obenko--{S}uris case.
\newblock In {\em Proceedings of the 4th Workshop ``Group Analysis of
	Differential Equations and Integrable Systems''}. 2009.
\newblock arXiv: 0902.3954.

\bibitem{MX}
Alexandre~V. Mikhailov and Pavlos Xenitidis.
\newblock Second order integrability conditions for difference equations: An
integrable equation.
\newblock {\em Letters in Mathematical Physics}, 104(4):431--450, 2013.

\bibitem{ASh}
V.E. Adler.
\newblock Toward a theory of integrable hyperbolic equations of third order.
\newblock {\em Journal of Physics A: Mathematical and Theoretical}, 45:395207,
2012.

\bibitem{AP1}
V.E. Adler and V.V. Postnikov.
\newblock On discrete 2d integrable equations of higher order.
\newblock {\em Journal of Physics A: Mathematical and Theoretical}, 47:045206,
2014.

\end{thebibliography}
\end{document}